\documentclass[acmsmall]{acmart}
\AtBeginDocument{%
  }
\setcopyright{acmlicensed}
\acmJournal{PACMMOD}
\acmYear{2023} \acmVolume{1} \acmNumber{4 (SIGMOD)} \acmArticle{268} \acmMonth{12} \acmPrice{15.00}\acmDOI{10.1145/3626762}

\usepackage[ruled,linesnumbered]{algorithm2e}
\usepackage{multirow}
\usepackage{array}
\usepackage{graphicx}
\usepackage{subcaption}
\begin{document}
\title{R2D2: Reducing Redundancy and Duplication in Data Lakes}\titlenote{All the work done in this paper was while all the authors were affiliated with Adobe Research.}

\author{Raunak Shah}
\authornote{Both authors contributed equally to this research.}
\email{raunaks3@illinois.edu}
\orcid{0000-0002-2889-7855}
\affiliation{%
  \institution{University of Illinois, Urbana-Champaign}
  \city{Champaign}
  \country{USA}
}

\author{Koyel Mukherjee}
\authornotemark[2]
\email{komukher@adobe.com}
\orcid{0000-0002-8690-323X}
\affiliation{%
  \institution{Adobe Research}
  \city{Bangalore}
  \country{India}
}

\author{Atharv Tyagi}
\orcid{0009-0008-3904-7953}
\email{athtyagi@adobe.com}
\affiliation{%
  \institution{Adobe Research}
  \city{Bangalore}
  \country{India}
}

\author{Sai Keerthana Karnam}
\orcid{0000-0003-1328-5167}
\email{skarnam@adobe.com}
\affiliation{%
  \institution{Adobe}
  \city{Bangalore}
  \country{India}  
}

\author{Dhruv Joshi}
\orcid{0009-0007-4543-846X}
\email{dhruvjoshi43@gmail.com}
\affiliation{%
 \institution{Indian Institute of Technology Kharagpur}
 \city{Kharagpur}
 \country{India} 
}

\author{Shivam Bhosale}
\orcid{0000-0002-5817-5271}
\email{sbhosale@adobe.com}
\affiliation{%
  \institution{Adobe}
  \city{Bangalore}
  \country{India}
}

\author{Subrata Mitra}
\orcid{0009-0009-8436-3119}
\email{subrata.mitra@adobe.com}
\affiliation{%
  \institution{Adobe Research}
  \city{Bangalore}
  \country{India}
}

\renewcommand{\shortauthors}{Raunak Shah et al.}

\begin{abstract}
Enterprise data lakes often suffer from substantial amounts of duplicate and redundant data, with data volumes ranging from terabytes to petabytes. This leads to both increased storage costs and unnecessarily high maintenance costs for these datasets. In this work, we focus on identifying and reducing redundancy in enterprise data lakes by addressing the problem of ``dataset containment". To the best of our knowledge, this is one of the first works that addresses table-level containment at a large scale.

We propose \textbf{R2D2}: a three-step hierarchical pipeline that efficiently identifies almost all instances of containment by progressively reducing the search space in the data lake. It first builds (i) a schema containment graph, followed by (ii) statistical min-max pruning, and finally, (iii) content level pruning. We further propose minimizing the total storage and access costs by optimally identifying redundant datasets that can be deleted (and reconstructed on demand) while respecting latency constraints.

We implement our system on Azure Databricks clusters using Apache Spark for enterprise data stored in ADLS Gen2, and on AWS clusters for open-source data. In contrast to existing modified baselines that are inaccurate or take several days to run, our pipeline can process an \textbf{enterprise customer data lake at the TB scale in approximately 5 hours} with high accuracy. We present theoretical results as well as extensive empirical validation on both enterprise (scale of TBs) and open-source datasets (scale of MBs - GBs), which showcase the effectiveness of our pipeline.
\end{abstract}

\begin{CCSXML}
<ccs2012>
<concept>
<concept_id>10002951</concept_id>
<concept_desc>Information systems</concept_desc>
<concept_significance>500</concept_significance>
</concept>
<concept>
<concept_id>10002951.10002952</concept_id>
<concept_desc>Information systems~Data management systems</concept_desc>
<concept_significance>500</concept_significance>
</concept>
<concept>
<concept_id>10002951.10002952.10003190</concept_id>
<concept_desc>Information systems~Database management system engines</concept_desc>
<concept_significance>500</concept_significance>
</concept>
</ccs2012>
\end{CCSXML}

\ccsdesc[500]{Information systems}
\ccsdesc[500]{Information systems~Data management systems}
\ccsdesc[500]{Information systems~Database management system engines}

\keywords{data management, data provenance, storage, redundancy}

\received{April 2023}
\received[revised]{July 2023}
\received[accepted]{September 2023}

\maketitle

\section{Introduction}
Enterprises nowadays ingest and process huge amounts of data daily.
Roughly 2.5 quintillion bytes of data are generated every day. The data storage unit market revenue worldwide in 2022 was 44.7 billion USD \cite{stat}. Managing such huge amounts of data is a tedious task, and it requires innovative data management techniques that are motivated towards saving costs. In fact, if not managed well, data ceases to be of value to enterprises or their customers or consumers and instead ends up incurring huge amounts of COGS (cost of goods sold) and liabilities for enterprises.

\textbf{\underline{Costs Associated with Data Regulations:}} 
GDPR\footnote{General Data Protection Regulation} and other privacy regulations (e.g., CCPA in California) are becoming increasingly important, with major tech firms incurring fines of hundreds of millions of dollars \cite{itgov-gdpr}. As a result, companies are spending huge amounts to ensure such compliance. According to estimates by IAPP\footnote{The International Association of Privacy Professionals} and others \cite{itif-privacy}, the spends range from \textbf{7.8 - 17 billion USD}, resulting in \textbf{~8\%} decline in profit~\cite{cepr}. 

The costs for data maintenance can be potentially reduced manifold by better management of data, namely, effective data retention and destruction policies \cite{pulse-gdpr}.

\textbf{\underline{The Redundant and Untracked Data problem:} }
To further complicate the problem, enterprise data lakes, with data volumes ranging in petabytes, often suffer from rampant data duplication. For instance, a marketer and an analyst from the same organization can process the same dataset, often in the same manner, for analytics and insight generation, saving the results for future use. This ends up creating even more data, without any record of the mutual relationship of the newly created datasets either with one another or with the existing datasets. As a result, enterprise data lakes often end up having untracked instances of related, or even, duplicate datasets. As an example, we analyzed the data containment for 3 customers in our enterprise data lake. We found that out of ~1400 datasets, over 115 datasets are fully contained within others, and 231 datasets are more than 75\% contained, accounting for 400 GBs of redundant data for these 3 customers alone.

\textbf{\underline{Cost Implications of Redundant Data:} }
Data storage as well as data access costs on the cloud\footnote{For example, the data storage and access costs on ADLS Gen 2 can be found here: \url{https://azure.microsoft.com/en-in/pricing/details/storage/data-lake/}.} both contribute substantially to the overall data maintenance costs for enterprises. Our enterprise data lake typically has at least one GDPR or privacy request-initiated access per customer dataset per week. This results in a \textbf{full table scan}\footnote{For many organizations, including ours, the data lake can only be accessed via APIs that use Apache Spark \cite{spark}, which \emph{does not use an indexed database}. In some cases, the data may be partitioned via timestamp, but that does not indicate user-level information. Thus, currently, a privacy-initiated access is a \textbf{full table scan}.} 
which is very expensive, especially since it is done at regular intervals. For a data lake of size ~1.4PBs in our enterprise, we can potentially save ~0.15 billion row scans incurred due to privacy-initiated (GDPR) accesses per month by deleting contained datasets.

Identifying and deleting redundant or duplicate data therefore becomes important. We define data redundancy as a condition in which the contents of a dataset are exactly contained in another dataset. We study this `Dataset Containment' problem across our data lake and propose to build a dataset containment graph, encoding the mutual containment information. Following that, we propose to delete some of the redundant (contained) datasets \textit{optimally}. 
Building a containment graph for an existing data lake is an interesting but hard problem that presents challenges in scalability as well as in identifying the correct relationships between datasets. 

Recently, several works have studied data discovery and data relatedness subject to different notions of similarity between datasets \cite{bharadwaj, josie, nextiajd, d3l, dataCivilizer, lcjoin, silkmoth, pexeso, ronin}. Many of these consider joinability or unionability \cite{tableunion, santos, union2} as the target relatedness metric, and often use schema based features to determine similarity in this context. These are insufficient for our use case of detecting containment since we need to look at content based similarity, as discussed in Section \ref{sec:schema-simi}. For similar reasons, the large body of works on entity resolution and matching, e.g. \cite{machamp, deepentity} also become inapplicable. Existing works on content similarity have looked at set similarity, which is insufficient for our case, where we need to detect exact row (record) level containment across subsets of columns (schema). Moreover, most of the existing work (on content similarity) has shown results on a smaller scale, for example, MBs to GBs, whereas we need to operate at a scale of TBs and PBs in enterprise data lakes. Works on storage layer de-duplication at block-level e.g. \cite{primarydatadedup, blockdedup1, blockdedup2} are also orthogonal to our use case because they would be ineffective in detecting containment arising out of logical operations, and inefficient due to additional storage and processing overheads.

We propose a scalable, hierarchical framework \textbf{R2D2} (which stands for ``Reducing Redundancy and Duplication in Data Lakes'') for identifying data containment relations, that operates on \textbf{TBs of data within a few hours with high accuracy.} 
To the best of our knowledge, this is one of the first works on identifying table-level containment at a large scale. Unlike R2D2, brute force approaches to compute ground truth and (modified) baselines from literature can take days to compute and/or give inaccurate results, respectively. We present results on datasets from multiple customer orgs\footnote{Here, a "customer org" refers to datasets sourced from and belonging to a particular customer in the data lake.} spread across different domains with different types and distributions of data to show the generalizability of our approach.
Additionally, we study the cost optimization problem arising out of the trade-off of retention versus deletion (and reconstruction on demand), along with the challenges of implementing such an algorithm in a dynamic, enterprise environment. 

We also discuss the challenges of identifying "approximately contained"  datasets in Section \ref{sec:approx-containment}.%}

\subsection{Our Contributions}
\begin{enumerate}
%\begin{enumerate}
    \item \textbf{R2D2 Framework: } We propose a scalable, hierarchical framework \textbf{R2D2} for identifying dataset containment relations, that operates on \textbf{TBs of data within a few hours with high accuracy}. It progressively reduces the search space by first building (i) a schema containment graph, followed by (ii) statistical min-max pruning, and finally, (iii) content level pruning. (Section \ref{sec:r2d2}).

    \item \textbf{Schema Graph Builder: } We propose an efficient schema clustering algorithm for building the schema containment relationship graph between datasets. We prove theoretically that no (ground truth) edges are missed in the resultant graph. (Section \ref{sec:sgb}).
   
    \item \textbf{Min-Max and Content Level Pruning: } We propose efficient algorithms based on statistical (minimum and maximum values in numerical columns) and content level similarity for eliminating edges from the schema graph to progressively build the dataset containment graph. We theoretically bound the sampling complexity with respect to the extent of containment with a probabilistic guarantee for correctly eliminating edges. (Section \ref{sec:mmp} and \ref{sec:clp}.)

    \item\textbf{Cost Optimization:} We provide an optimization algorithm that minimizes the overall expected storage and access costs while respecting latency constraints. It takes a containment graph as input and identifies redundant datasets that can be deleted (and reconstructed on demand). (Section \ref{sec:optimization}).

    \item \textbf{Empirical Validation: } We validate our pipeline through substantial empirical results on both enterprise and open-source data, ranging from MBs to TBs. We compare with baselines and show that our method achieves very good results compared to the ground truth, both in terms of correctness and efficiency. We also test the scalability of the pipeline and potential storage and compute savings. (Section \ref{sec:experiments}).
\end{enumerate}

Figure \ref{fig:framework} shows an end-to-end high level view of our proposed pipeline and framework.
\begin{figure}[htbp]
    \centering
    \includegraphics[width=0.7\linewidth]{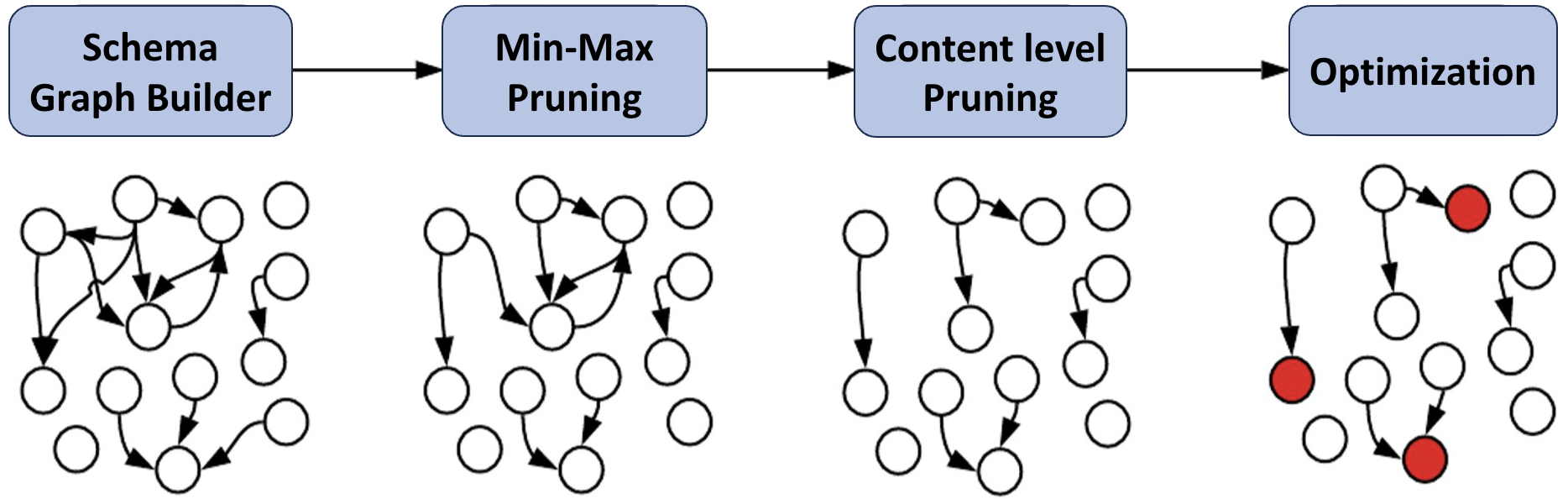}
    \caption{An illustration of our end-to-end framework. After computing the schema graph, each step successively prunes the graph space, first based on the minimum and maximum values of columns and then based on the contents. At the end, nodes marked in red are recommended for deletion.}
    \label{fig:framework}
\end{figure}

\subsection{Enterprise Schema and Data Containment}
\label{sec:schema-simi}
In our enterprise data lake, different customer orgs have different distributions of schemas as well as data. In some orgs, many pairs of datasets have similar schema, whereas in others, this number is lower (refer Fig. \ref{fig:schema_cont}). Not only does the distribution of schema similarity vary widely across customers, but the schema similarity alone is insufficient to conclude data similarity or extent of containment. Existing approaches such as \cite{bharadwaj} that focus on identifying similar datasets with respect to joinable columns, use schema based features and data features computed on small samples of data with high accuracy. However, in our enterprise data lake, tables with similar, generic schemas often have very different distributions of values within a column, since the tables often come from different sources and have gone through varied types of processing and transformations before reaching their current state. In another experiment, we considered tables with the same schema from a customer's data, and computed quantiles (at fractions 0, 0.5, 0.8, 0.95, and 1) from the distribution of values in each column. We found that despite having the same schema, over 20\% of table pairs have normalized quantiles that are at least 50\% different, which indicates a large degree of independence between the schema (column) name and exact values within a column. Thus, it is important to consider data based features as well as schema based features while identifying containment. The scale of our problem is massive, and building the ground truth dataset containment graph through pairwise dataset comparisons would require the order of $10^{21}$ pairwise record level comparisons (see Table \ref{tab:row-ops}) for certain customer accounts.

\begin{figure}[htbp]
  \centering
  \begin{subfigure}{0.48\textwidth}
    \includegraphics[width=0.95\linewidth]{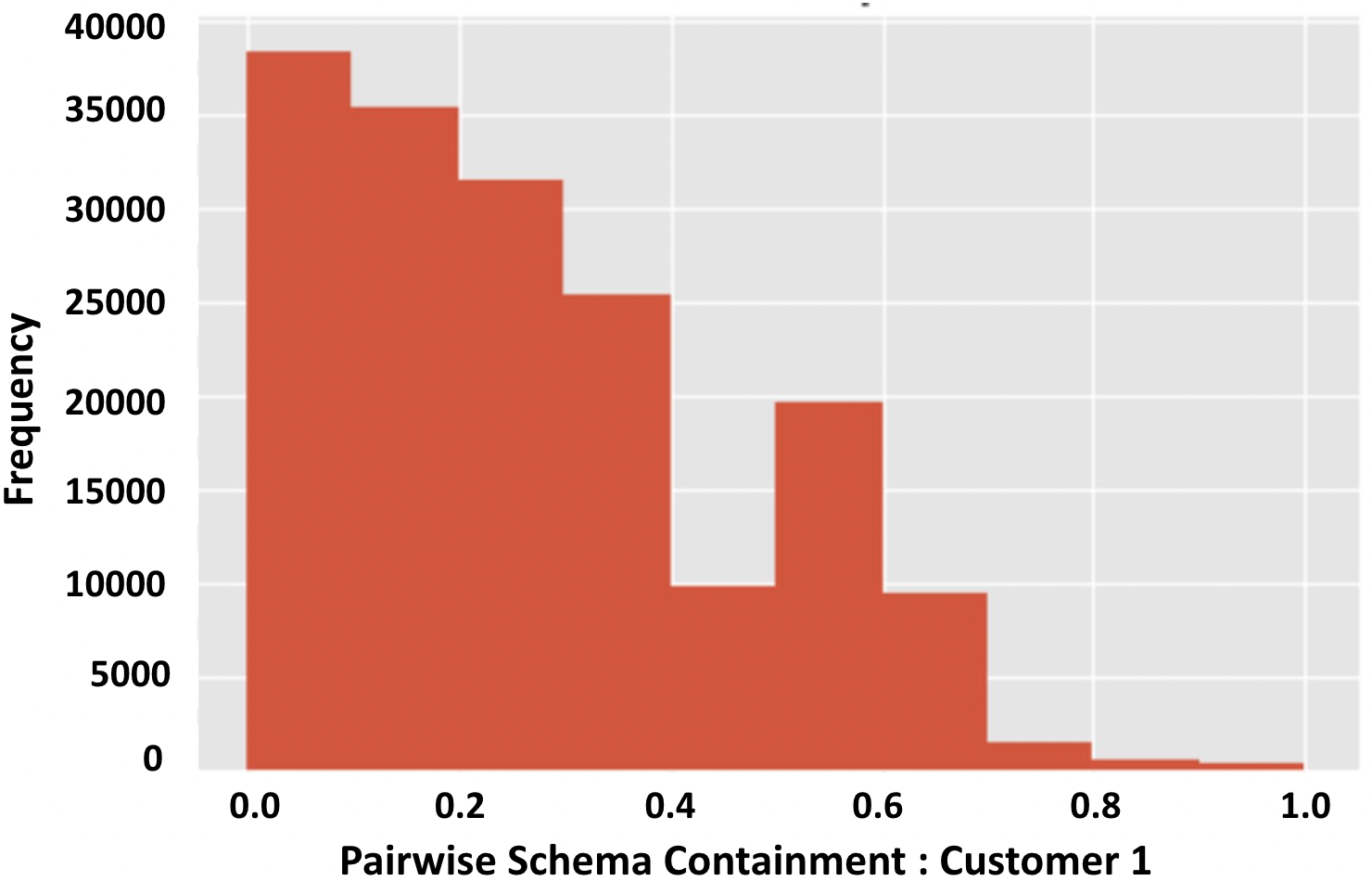}
  \end{subfigure}
  \hfill
  \begin{subfigure}{0.48\textwidth}
    \includegraphics[width=\linewidth]{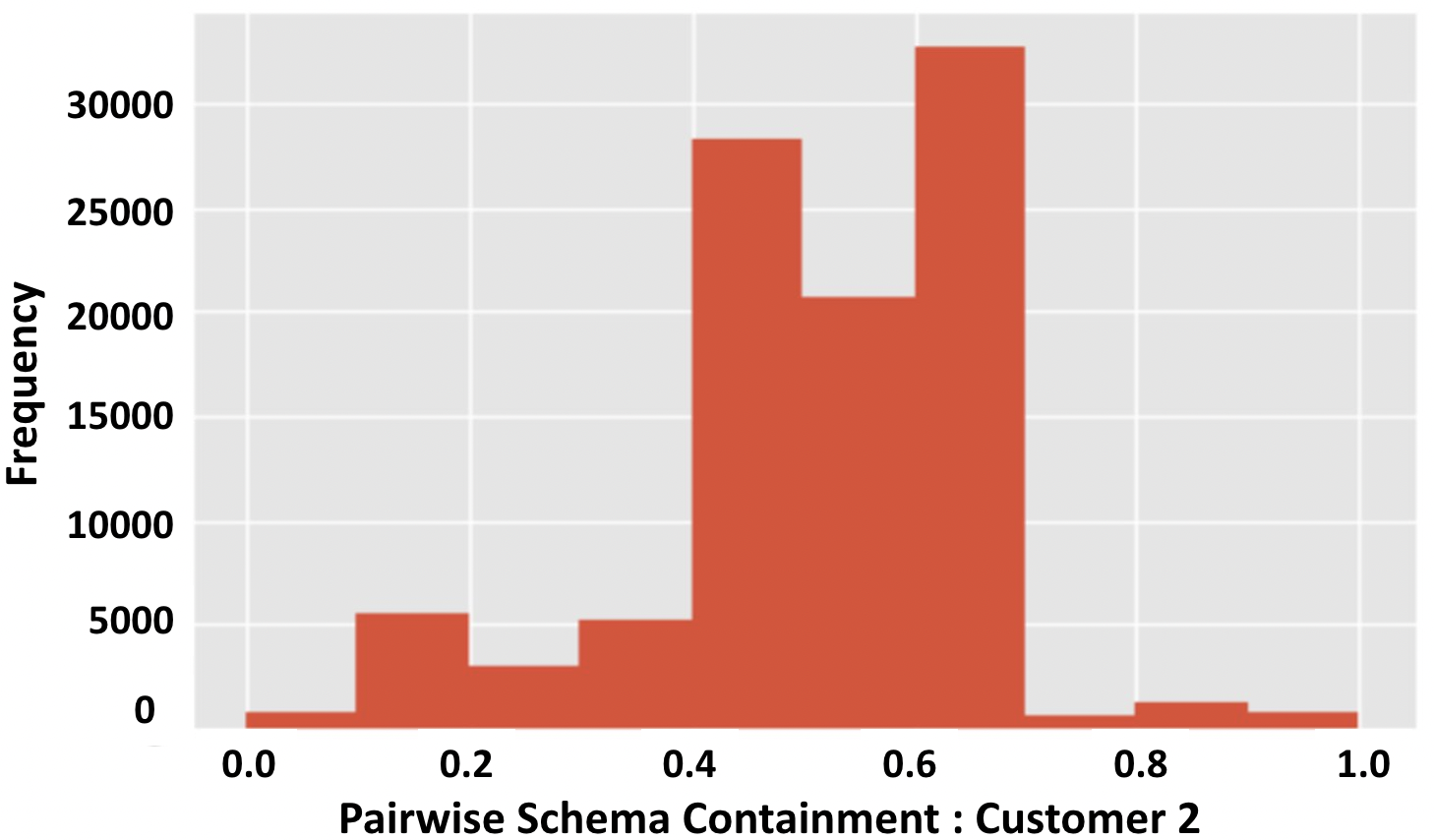}
  \end{subfigure}
  \caption{Histograms of schema containment between dataset pairs in 2 different customer orgs in our enterprise data lake. The x-axis ranges from 0 to 1 (zero schema intersection to complete containment between two tables). Clearly, the distributions of schema similarities vary across orgs.}
  \label{fig:schema_cont}
\end{figure}

\section{Related Work}
\label{sec:related}
\textbf{Joinability Discovery: } Data discovery based on joinability has been widely studied in the literature \cite{nextiajd, josie, bharadwaj, pexeso}. In general, joinability based approaches do not apply to the containment problem, as joinability is between two column pairs and not the complete set of columns from tables. It is possible for two tables from completely different domains to be joinable but they may not be related at all. We next discuss a few existing works in this field.

NextiaJD \cite{nextiajd} studies joinability using a learning-based approach, only considering attributes comprised of string datatypes (ignoring numerical attributes). The quality of join is computed based on containment (Jaccard Distance) and cardinality proportions between two attributes (columns). Similar to \cite{nextiajd}, JOSIE \cite{josie} too focuses on non-numerical attributes. JOSIE aims to find tables that can be joined with the given column on the largest number of distinct values, changing columns to sets. JOSIE minimizes the cost of set reads and uses inverted index probes to find the top-k related sets. However, building an inverted index is inefficient, both computationally and from the perspective of memory management. Moreover, considering columns as sets loses crucial row tuple based information which is important for identifying containment\footnote{For example, consider 2 tables with 2 columns. Table 1 entries are: (June, 20) and (May, 12). Table 2 entries are: (June, 12) and (May, 20). At a set level, there is column-wise containment in both columns, but as row tuples, there is no containment.}. Bharadwaj et al.\cite{bharadwaj} address the task of finding joinable tables, that occur in the join clauses in enterprise queries. They train a random forest classifier to predict whether two columns are related or not based on features such as metadata (column name similarities), and sample-based features (based on similarities of the first 1000 values in numeric columns). They find that a very large proportion of these joinable columns have similar schemas, which makes schema features useful for identifying joinability. However, as discussed earlier in Section \ref{sec:schema-simi}, this is not enough to judge the containment between tables. Additionally, we need to consider both numeric and non-numeric features.

\textbf{Inverted Index based approaches: }
Inverted index based data structures have been widely used to address the problem of set similarity, column joinability, and relatedness, e.g. \cite{silkmoth, setSimTheoreticalGuarantees, dataCivilizer, lcjoin, lshensemble, allign, josie}. These approaches primarily reduce the number of column/set comparisons. They do not optimize how the table similarity is computed and generally require full sweeps of all the rows in a table. As a result, such approaches do not scale well with increasing dataset sizes. We now discuss a few of these works.

LSHEnsemble\cite{lshensemble} identifies containment between domains (or sets of column attributes) of tables on the web. Although the number of domains can be very large (and this is what they focus on), identifying containment between two attributes is not expensive since domain sizes are relatively small as they are mainly keywords. This is not the case when we want to find the containment between entire columns. In our case, the number of datasets (equivalent to the number of domains) is less of a bottleneck, and the number of rows (equivalent to domain sizes) is extremely large, so techniques based on min-hash become computationally infeasible.
LCJoin\cite{lcjoin} tries to find all pairs of sets (R, S) such that $R \subset S$ where R and S belong to a collection of sets $\mathbb{R}$ and $\mathbb{S}$ respectively. Although we can apply \cite{lcjoin} to find subset relationships between column sets of different tables, it does not translate to overall table containment, for similar reasons as explained earlier. Similarly, Silkmoth \cite{silkmoth} tries to solve the problem of set relatedness and set containment. The metric used measures the relatedness of two sets by treating the elements as vertices of a bipartite graph and calculating the score of the maximum matching pairing between elements. Unfortunately, the metric suffers from expensive computational cost, taking $O(n^3)$ time, where n is the number of elements in the sets, for each set-to-set comparison. For data lakes, the time complexity of $O(n^3)$ is unacceptable because of the size (number of rows) in enterprise datasets, and thus \cite{silkmoth} cannot be applied in our setting. 
%More specifically, the major bottleneck for us is the time taken to find similarity between two datasets of large size, rather than the number of comparisons made between the datasets.
% \textcolor{blue}{compare with data civilizer if time permits}

\textbf{Data Discovery: }
There are other approaches that look into discovering related datasets, e.g. \cite{d3l, juneau, ronin, optimOrgDataLake}.
D3L \cite{d3l} uses schema-and instance-based features (name, value, format, embeddings, and domain/distribution) to construct hash-based indices (locality sensitive hashing based) that map these features into a uniform distance space. This makes it possible to consider hash value similarities as measures for table relatedness. However, LSH indexing is computationally expensive, and dynamically adding a new dataset requires recomputing the indices over the entire space of data, making it computationally infeasible. 
RONIN \cite{ronin} enables user exploration of a data lake by navigation of a hierarchical structure. Here, table similarities are computed by averaging word embeddings of tokens in the table schema and metadata (description, tags) where available. This does not consider exact containment of schema nor does it consider content level similarity at all. Juneau\cite{juneau} finds target tables related to a source table, given a task such as finding tables with some augmented data, linked (joinable) tables, and finding tables where data cleaning has occurred (like removing null values). Their use-case is focused on small tables used by data scientists in a jupyter notebook environment, and evaluation is done on a 5GB corpus with an average table size of 1MB. Our use case however focuses on finding redundant data in a vast data lake at a much higher scale in terabytes. Additionally, they assume the presence of a data provenance graph, which is not given in our setting.

\textbf{Data Versioning: }
There has been some work \cite{dataversioning, orpheusDB} on selectively storing and deleting dataset versions. Although we are not dealing with dataset versions, it is useful to examine this related line of work. \cite{dataversioning} takes a static provenance graph as input, and trades off the storage and reconstruction cost of keeping or deleting dataset versions. This problem is related to the optimization part of our pipeline however the exact optimization formulation considered by us is different from \cite{dataversioning, orpheusDB}. 

\textbf{Storage Layer Deduplication: }
Some existing approaches de-duplicate data via block-level de-duplication at the storage layer\cite{primarydatadedup, blockdedup1, blockdedup2}. 

As an enterprise data management platform, we operate above the abstractions provided by Apache Spark \cite{spark} and therefore we do not have access to the storage layer nor do we have access to the individual files. In an enterprise data lake the sources of data are largely heterogeneous, so a system that de-duplicates at write-time would be infeasible. Even for post-processing de-duplication, the block sizes used in existing work are generally very small (of the order of KBs). This implies maintaining indexes for up to hundreds of millions of blocks since the size of even a single table can go up to tens or hundreds of terabytes. Storing such metadata in a ``chunk store", and/or running block compression is infeasible in our enterprise setting, and would take too long. On the other hand, \textbf{R2D2 can be run end-to-end as a Spark job in memory}. The only extra metadata we use is existing partition level metadata that already exists. It is also convenient to delete datasets or partitions completely - deleting chunks from certain parts of partitions and then reconciling these partitions would require maintaining such a system at the storage layer, which is difficult (we use cloud services for storage). Finally, chunking cannot handle logical operations/transformations - e.g. it would consider a sorted and unsorted table to be different. However, these tables are contained within one another. In fact, in Spark, since row ordering is not preserved, the tables are essentially equivalent. Our pipeline would be able to capture this.

\section{Problem Definition}
\label{sec:problem-statement}
Our goals are as follows. 
We first want to generate a mapping of the redundancy or, containment among the datasets in the data lake in a scalable, efficient manner. Second, we want to identify candidates for ``safe deletion'' from this pool of datasets to minimize the redundancy and associated storage and maintenance costs. We use the terms `datasets' and `tables' interchangeably in the rest of the paper because we work with tabular datasets (commonly used for storing digital transactions and clickstream event logs). We now formally define our problem statement.

We first define the term \textit{``containment fraction''} of one table in another. 
The containment fraction of a table $A$ in table $B$  is defined as $CM(A, B) = \frac{|A\cap B|}{|A|}$, where $n(B) \geq n(A)$. Note that we consider containment of table schema as well as full table contents. We use the same notation for both these cases.

If A and B are schemas, $n(B)$ refers to the length of the flattened schema set in B, and $|A\cap B|$ refers to the length of the intersection between the flattened schema sets. If they are tables, $n(B)$ refers to the number of rows in B and $|A\cap B|$ refers to the number of rows common to both tables. Throughout the rest of the paper, $A \subseteq B$ will denote that A is contained within B (i.e. $CM(A, B) = 1$). We will use this notation for both schema as well as table level containment. 
Next, we explain the term \textit{"safe deletion"}. A recommendation of deletion for a candidate dataset $D$ is considered to be ``safe deletion'' as long as (i) the (recommended) retained set of datasets contains at least one parent dataset $D_p$ from which $D$ can be reconstructed if accessed in the future, and (ii) this reconstruction can be done within a bounded latency (this will be further discussed in Section \ref{sec:optimization}.1) This is to ensure that the Quality of Service (QoS) experienced by the customer is high and as per the Service Level Agreements (SLAs).

We study the following two problems in this work:

\textbf{1. Containment Detection:} Identify all pairs of tables $(P, Q)$ such that $CM(P, Q) \geq T$, where $n(Q) \geq n(P)$. In this paper we focus on solving the problem of exact containment detection, i.e. when $T = 1$, which implies $P \subseteq Q$, or that all the rows in $P$ are present in $Q$. Note that this would also capture exact duplicates ($P = Q$). Approximate containment detection, i.e. when $T < 1$, introduces several additional challenges. These are out of the scope of this paper - however, we briefly discuss this case in Section 7.3.

\textbf{2. Recommending Candidates for Deletion:}  Based on the containment information, identify all tables that can be safely deleted, without any loss of information, so that the overall costs incurred due to data storage, data access as well and data maintenance are optimized. In order to formalize the above statement, first consider a dataset relationship graph, where the datasets are nodes, and there exists an edge directed from dataset node $B$ to $A$ if and only if : (a) $A$ is contained in $B$, that is, $CM(A, B) = 1$, (b) the transformation to generate $A$ from $B$ is known to the system, and (c) the expected latency to generate the transformation is bounded and within some threshold. For each such directed edge from $B$ to $A$, we denote $B$ as a parent of $A$ and we can reconstruct $A$ from $B$, if required by a customer-initiated access, in case $A$ is deleted. The problem is to identify the optimal set of such contained datasets for deletion while retaining at least one parent for every deleted dataset. This captures the ``safe deletion" requirement. Detailed problem descriptions along with important system level aspects for the deletion process are discussed in Section \ref{sec:dynamic-updates}.

\section{R2D2 Framework}
\label{sec:r2d2}
In this section, we describe the key components of the R2D2 framework. 
We approach the problem in a step-by-step hierarchical manner, which makes our method scalable as well as modular. We model the underlying space of datasets as a graph, capturing their pairwise relationships. We solve the data redundancy problem in two parts. First, we identify data containment between pairs of tables and build the corresponding containment graph. Second, we solve an optimization problem on top of this containment graph, that recommends the subset of datasets that can be ``safely deleted'' (as explained in Section \ref{sec:problem-statement}), while minimizing the overall storage, maintenance, and expected reconstruction costs. In order to solve the first problem, that is, building a containment graph, we approach the problem hierarchically, first building a schema based containment graph, then progressively eliminating edges from the graph to build the containment graph, where the existence of a directed edge indicates with high probability that the child (where the edge is incoming) is contained in the parent (where the edge is outgoing). 
In this work, we assume that exact schema level containment is a necessary condition for full containment. We plan to study approximate schema and content containment in future work. However, we do discuss aspects of approximate containment and some special cases that we can handle in Section \ref{sec:approx-containment}.
We next discuss the different steps of the R2D2 pipeline in detail.

\subsection{Schema Graph Builder: Computing Schema Containment Graph}
\label{sec:sgb}
At this stage, our goal is to build a schema containment graph that might contain additional edges (to be pruned at later stages) but \emph{no missing edges}. In other words, our goal is to ensure 100\% recall for detecting containment edges at this stage. Here, an edge $B\rightarrow A$ denotes the presence of a pairwise schema level containment ($A.schema \subseteq B.schema$). Note that schema level containment is a necessary, but insufficient condition for full table level containment. One possible brute-force approach would be to do pairwise comparisons between all pairs of schemas. For $N$ datasets (where $N$ is a very large number for even moderately sized enterprise data lakes), this would result in $O(N^2)$ comparisons and no additional edges. However, we propose to trade-off precision (without affecting the recall, that is, not missing any edges) with compute at this stage. 

We propose to find an initial set of (overlapping) clusters in the space of schema entities, and build the graph by examining the containment between any pair of schema entities that are members of the same cluster. If containment exists, we construct a directed edge directed from the larger schema to the smaller schema. Instead of using traditional clustering algorithms such as \textsc{K-Means} which would be computationally expensive, and would require featurization, we propose the following algorithm for generating the overlapping clusters. We first select `centers' in the schema entity space, and then assign cluster members to each center based on containment. Note that the center is also a cluster member. We call this algorithm \textsc{Schema-Graph-Builder} or \textsc{SGB} in short. 

We next describe the \textsc{SGB} Algorithm in detail.

\begin{enumerate}
\item First, construct a schema set for each dataset. For flat schemas, this is simply a list of all columns. For tree schemas, which are typical in enterprise workloads, the schema set is computed by flattening the schema tree so that the resulting tokens are distinct. For example, a schema tree with root $product$ and children $price$ and $id$ would be represented as $product.price$ and $product.id$. Empirically we have observed that this can be done in seconds.

\item Sort the list of schemas in non-increasing order of the size or cardinality of the corresponding schema sets.

\item Initialize the list of schema centers as empty sets.

\item Traverse the list of schemas in the sorted order.

\item For the next schema set, check if it is contained in any of the current centers.  
If not, then this becomes a new center. 
Else, add it as a cluster member to \textbf{each} center that it is contained in. Continue traversing (i.e., go to Step 4) till the end of the list.

\item Now, add edges between every pair of schemas within a cluster that satisfy the exact containment condition. We include the cluster center schema as a part of a cluster in this process. This builds the Schema Containment Graph. Note that one entity can have edges to members of multiple clusters. This is possible because the same node can be a member of two or more clusters.
\end{enumerate}

\begin{figure}[htbp]
    \centering
    \includegraphics[width=0.7\linewidth]{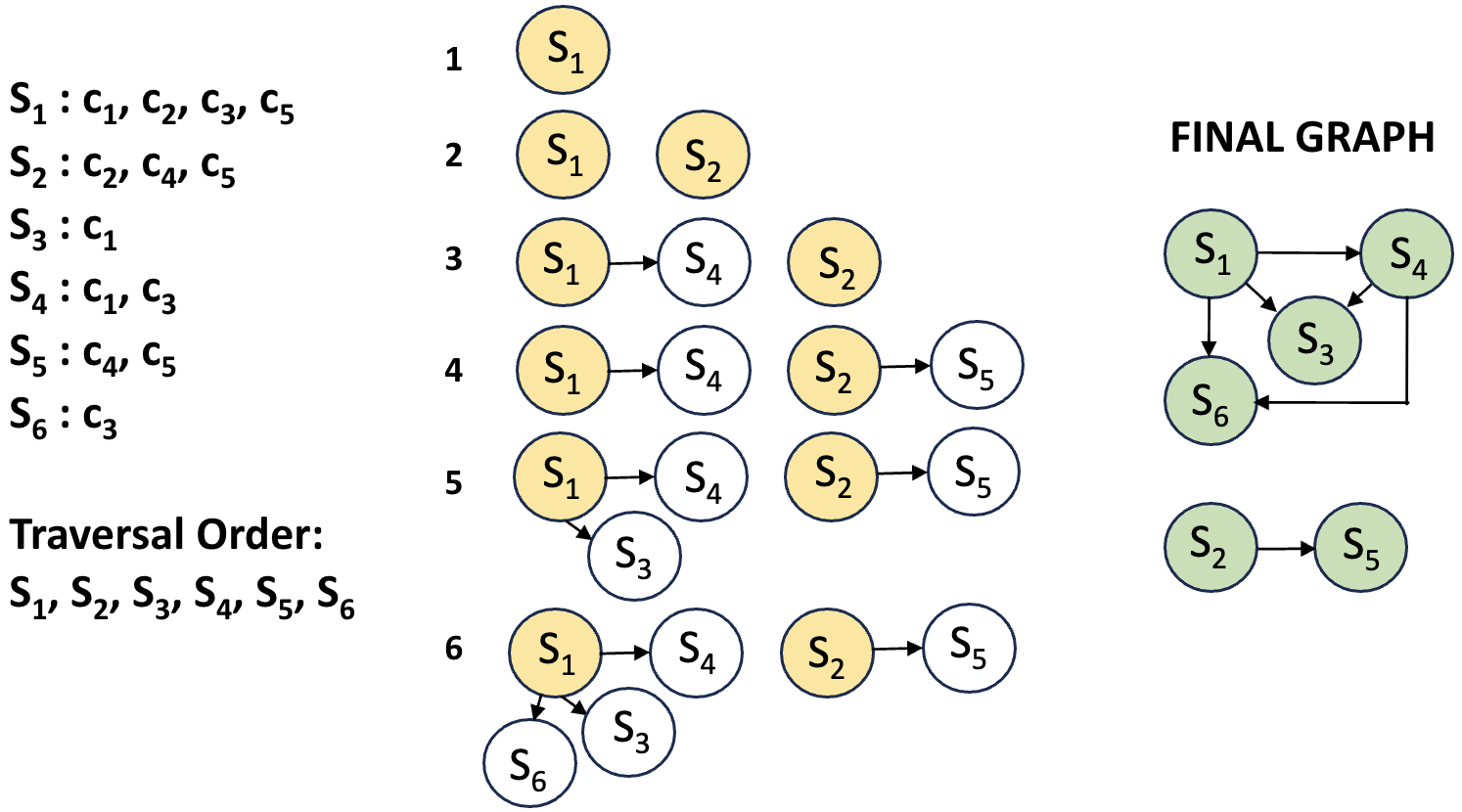}
    \caption{$S_i$ - schemas, $c_i$ - columns. Left: 6 schema definitions along with our order of traversal (non-decreasing order of size of the schema) in the SGB algorithm. Middle: Populating clusters in each iteration in the SGB algorithm. Cluster centers are in orange. Right: How we construct the final schema graph from these clusters.}
    \label{fig:schema-containment}
\end{figure}

Figure \ref{fig:schema-containment} illustrates the above steps on a small example in a step-by-step manner. The flattened columns in each schema along with the traversal order are shown on the left (steps 1-4). Each iteration wherein we populate the clusters and assign cluster centers is shown in the middle of the figure (step 5). Finally, we add edges within clusters to create the final schema graph based on whether containment exists or not (step 6). The pseudocode is given in Algorithm \ref{alg:sgb}.
For $K$ clusters, the total complexity of the sorting and traversal would be $O(N\log{N}) + O(K(N-K))$. 
We next prove that we will not miss any edge by \textsc{SGB} (even though many additional edges might be detected at this stage). 

\begin{algorithm}[htbp]
\caption{SGB: Schema Graph Builder} \label{alg:sgb}
\KwData{$S$ = list of schemas, $C$ = set of clusters}
\KwResult{$G = \textrm{Schema Containment Graph}$}
$cluster = Struct(center, members)$\\
$schema = Struct(size, dataset)$\\
\For{$schema\; s \in sorted(S, by=size, descending)$}{
    $contained = 0$\\

    \For{$cluster\; c \in C$}{
        $cc = c.center$\\
        \If{$(s.size <= cc.size) \;\textrm{\textbf{and}}\; (s\subseteq cc)$}{
            $c.members.append(s)$\\
            $contained=1$\\
        }
        
    }
    
    \If{$contained == 0$}{
        $newCluster = \mathbf{new}\; cluster(s)$\\
        $C.append(newCluster)$
    }  
}
\For{$cluster\;c \in C$}{
    \For{$distinct\;pairs (x,y) \textrm{ s.t. } x,y\in c.members$}{
        \If{$y \subseteq x$}{\tcp{WLOG x.size > y.size}
        $G.addEdge(x.dataset \rightarrow y.dataset)$ 
        }
    }
}
return G
\end{algorithm}

\begin{theorem}
    The schema containment graph identified by the algorithm \textsc{SGB} 
    will not be missing any correct (containment) edges. 
\end{theorem}

 \begin{proof}
    Let us assume for the sake of contradiction that the schema graph built by SGB for a given data lake has missed at least one correct edge $B \rightarrow A$, that is, $A.schema \subseteq B.schema$. That means when $A.schema$ was being processed or traversed in the list of schemas, the containment with $B.schema$ was undetected. Since $B.schema$ is longer than $A.schema$, it would have been traversed earlier in the list. First, consider the case that $B.schema$ is a cluster center. Since the algorithm checks containment with each center explicitly, this edge would have been detected, leading to a contradiction. Now, consider the case that $B.schema$ is not a cluster center. Then, by definition of the algorithm, $B.schema$ must have been marked as a cluster member of at least one existing cluster, say $C.schema$. This further implies that $B.schema$ is fully contained in $C.schema$. Therefore, when $A.schema$ is being traversed by SGB, it would have detected containment to at least one cluster center $C.schema$ and would have marked $A.schema$ as a cluster member of $C.schema$. However, by definition of the algorithm, we build an edge between every pair of cluster members, which would include an edge between $A.schema$ and $B.schema$ since both are members of the same cluster. This again leads to a contradiction. This completes the proof. 
 \end{proof}

\subsection{Min-Max Pruning}
\label{sec:mmp}
The second step of our algorithm (Algorithm \ref{alg:mmp}) takes the schema graph as input and chooses which edges to prune. Note that our final goal is to have a graph where a directed edge $B\rightarrow A$ between two nodes (datasets) represents that $A \subseteq B$, i.e. $A$ is contained within $B$. In this step, we exploit the relationship between the minimum and maximum values of columns in datasets that are contained within one another. Assume $n(A) <= n(B)$. Consider $C$ to be the set of common columns between $A$ and $B$. If $A \subseteq B$, it is necessary for (i) $\min A.c \geq \min B.c \; \forall\; c \in C$ and (ii) $\max A.c \leq \max B.c \; \forall\; c \in C$. If any one of these conditions is violated for any column, we can eliminate that edge and safely conclude that containment is not present. Moreover, for datasets that are partitioned and stored in parquet format, values such as the columnar minimum and maximum are often stored as metadata. This makes looking up these values very fast, since a full table scan is not necessary - looking at partition level metadata would suffice. Caching the columnar minimum and maximum is another option that would improve the speed even further.

\begin{algorithm}[h]
\caption{MMP: Min-Max Pruning} \label{alg:mmp}
\KwData{$G$ = Schema Containment Graph}
\KwResult{$G$ = Updated Containment graph}
\For{$edge\; x\rightarrow y \in G.edges$}{
    $commonCols = x.schema \cap y.schema$\\

    \For{$c \in commonCols$}{
        \If{$((x.c.min > y.c.min) \textrm{ \textbf{or} } (x.c.max < y.c.max))$}{
        $G.removeEdge(x\rightarrow y)$\\
        \textbf{break}
        }
    }
}
return G

\end{algorithm}
\subsection{Content Level Pruning}
\label{sec:clp}
Finally, we prune the graph based on the contents of the datasets under consideration. We exploit the insight that if table level containment holds, i.e. $A \subseteq B$, it will also hold between a sample of $A$ and $B$, i.e. $sample(A) \subseteq B$.
For instance, consider a categorical or timestamp type of column, such as "id" or "timestamp", which are commonly found in enterprise data. Consider the rows where specific values are present in such columns - this is equivalent to running a query such as \verb|SELECT * FROM A WHERE column = value|. If any of the rows returned by this sample are not present in $B$, we can safely conclude that containment does not hold between $A$ and $B$ either. Then we can prune the edge $B\rightarrow A$ from the graph.
Choosing to sample values from $A$ that result in a lower number of rows will lower the time complexity of checking containment with $B$. Thus, this idea can be extended further - instead of running a query with a single \verb|WHERE| filter, we can simultaneously use multiple \verb|WHERE| filters on different columns and values to make this step even faster. 
Another possible extension is to sample from both $A$ and $B$ using a query with \verb|WHERE| filters, and then check containment between the sampled rows. For example, if $s_A$ = all the rows of A with a particular timestamp $t$ and $s_B$ = all the rows of B with timestamp $t$, containment should be preserved between them, i.e. $s_A \subseteq s_B$ if $A \subseteq B$. Otherwise, we can conclude $A \nsubseteq B$. This would hold as long as we sample using \verb|WHERE| queries, since the containment relationship is unaffected by such a sampling operation. The key in all of the above variations is that the sampling does not need to scan the full table. If information about the column to be sampled is known beforehand (e.g. we may know a range of timestamp values present in the dataframe), or the data has been partitioned based on timestamp values, or if the database is indexed (all of which are common scenarios), only certain rows of the table need to be accessed while sampling. The pseudocode is given in Algorithm \ref{alg:clp}. 

\begin{algorithm}[h]
\caption{CLP: Content Level Pruning} \label{alg:clp}
\KwData{$G$ = Containment Graph after MMP, $s$=max columns to consider, $t$=max rows to sample}
\KwResult{$G$ = Updated Containment graph}
\For{$edge\; x\rightarrow y \in G.edges$}{
    $commonCols = x.schema \cap y.schema$\\
    $searchCols = commonCols.sample(s)$\\

    $sY = sample(y, searchCols, maxRows=t)$\\
    $combined = sY.join(x, type="left-anti")$\\
    \If{$!combined.isEmpty$}{
        $G.removeEdge(x\rightarrow y)$
    }
}
return G

\end{algorithm}

Let us now examine the number of samples that would be required to give a probabilistic guarantee on the correctness of containment inference, similar to a PAC bound. 
\begin{theorem}
    Given a pair of datasets with fraction of containment at most $(1 - \epsilon)$, the number of samples $n_s$ required to ensure that we are able to eliminate or prune the edge between the pair of datasets with probability at least $(1 - \delta)$ is $n_s \geq \ln{\frac{1}{\delta}}/\ln{\frac{1}{1 - \epsilon}}$.
\end{theorem}

\begin{proof}
Consider two datasets $D_1$ and $D_2$, where the fraction of containment of $D_2$ in $D_1$, that is, $CM(D_2, D_1) \leq (1 - \epsilon)$. 
Ideally, we want to be able to prune the edge between $D_1$ and $D_2$ with a high probability $P(D_2 \notin D_1) \geq (1 - \delta)$, given the containment is $\leq (1 - \epsilon)$. 
This requires that at least $1$ sample be retrieved from the non-intersecting rows of $D_2$ and $D_1$. Let the sampling be uniformly random with replacement and let the minimum number of samples required be $n_s$. We denote the part of $D_2$ contained in $D_1$ as $D^c_2$.
\begin{align} 
    P(D_2 \notin D_1) &= 1 - P(D_2 \subseteq D_1) \ = 1 - \frac{|D^c_2|^{n_s}}{|D_2|^{n_s}}\ = 1 - {\left(\frac{D^c_2}{D_2}\right)}^{n_s}
\end{align}
The containment $\frac{D^c_2}{D_2} \leq (1 - \epsilon)$. We want the probability $P(D_2 \notin D_1)$ to be $\geq (1 - \delta)$ for some small $\delta$. Substituting, we require
\begin{align}
    P(D_2 \notin D_1) &\geq 1 - (1-\epsilon)^{n_s} \ \geq (1 - \delta)
\end{align}

Therefore, we need $\delta \geq (1 - \epsilon)^{n_s}$. This gives 
$n_s \ln{\frac{1}{1-\epsilon}} \geq \ln{\frac{1}{\delta}}$. 
or, $n_s \geq \frac{\ln{\frac{1}{\delta}}}{\ln{\frac{1}{1 - \epsilon}}}$. 
As an example, for $\delta = 0.05$ (that is, a probability of 95\%) and fraction of containment at most $0.9$ or 90\% (that is, $\epsilon = 0.1$), $n_S \geq 29$.
\end{proof}

\section{Minimizing Redundancy}
\label{sec:optimization}
Thus far, we have presented the pipeline to construct a dataset relationship graph, where the relationship denotes \emph{exact} containment. Specifically, if a dataset $q$ is completely contained in another dataset $p$, then with high probability, there exists a directed edge from the node $p$ to the node $q$ in this graph. Now, we would like to exploit this containment information to reduce the storage, maintenance, and related compute costs. We propose to ``safely delete'' some of the contained datasets, while retaining the others, thereby lowering the overall costs.

\subsection{Graph Pre-processing for  ``Safe Deletion''}
We first pre-process the dataset relationship graph to encode the reconstruction (or, transformation) information between datasets. This is required to ensure the ``safe deletion'' requirement. For each directed edge from dataset $p$ to dataset $q$, we first ensure that the \emph{transformation} used to generate the child dataset $q$ from the parent dataset $p$ is known. Such transformation can be potentially known through human input\footnote{Human examination of every edge is feasible at the current step of the pipeline, as the number of edges to consider has been reduced greatly at this stage. For enterprise data, we have empirically observed that the number of edges remaining at this step is of the order of ~100-300 across customers.} or inferred by a model. 
%This transformation information is now associated with the corresponding edge. 
We propose to use human input because of the sensitive and critical nature of the client data. Inference through a model is currently a work in progress and is out of scope for this paper. In case the transformation between a pair of datasets is not known or cannot be inferred by the human expert, we prune the corresponding edge.
Let us assume that the transformation for the edge $e$ from $p$ to $q$ has been given by the human expert. We need to estimate the (monetary) reconstruction cost and latency associated with such a reconstruction. The reconstruction cost mainly comprises of the read cost of the parent dataset $p$\footnote{Generally, enterprise use cases do not require heavy compute for data analytics, hence the compute charges are negligible.} and the write cost of the child dataset $q$. The read and write costs can be easily estimated by the size of the parent and child datasets respectively and the read cost and write cost per unit size of data as charged by the corresponding cloud provider. Let the read cost per unit size of data be $r$ and the write cost per unit size of data $w$. Then the read cost is estimated as: $r \cdot s_p$, where $s_p$ is the size of the parent dataset $p$ in bytes, and the write cost is estimated as: $w \cdot s_q$ where $s_q$ is the size of the child dataset $q$ in bytes. Therefore, the estimated total cost $C_e$ of reconstructing $q$ from $p$ is given by: $C_e \approx r \cdot s_p + w \cdot s_q$. Generally, the cloud costs for write operations in the premium and hot tiers are an order of magnitude higher than the read costs\footnote{\url{https://azure.microsoft.com/en-in/pricing/details/storage/data-lake/}.}, hence if the sizes $s_p$ and $s_q$ are comparable, then the cost $C_e$ can be further approximated by the write costs alone: $C_e \approx w \cdot s_q$.% }

Note that ``safe deletion'' does not simply ensure that a dataset can be reconstructed (if deleted) on demand, it also ensures that the associated latency of such a reconstruction is bounded to maintain the Quality of Service (QoS) experienced by the customer. Hence, if the estimated latency exceeds that QoS bound, that edge needs to be deleted. 

The read latency of $p$ and the write latency of $q$ can be estimated from the historical logs or directly from the cloud provider's latency guarantees. For estimating from historical logs, a simple approach would be to average the write (read) latency values normalized by the corresponding dataset sizes on which the write (read) operations were performed as recorded in the historical logs, and multiply it by $s_q$ ($s_p$). 
Let the estimated latency normalized by byte sizes for write be $w_\ell$ and for read be $r_\ell$.
The total latency of edge $e$, $L_e$ is therefore estimated as $L_e \approx r_\ell \cdot s_p + w_\ell \cdot s_q$. 
Now, if $L_e\geq Th$, where $Th$ is a latency threshold determined by QoS requirements by clients, then we remove $e$.

\subsection{Opt-Ret: Optimally Retain Datasets to Minimize Expected Costs}
We now present an optimization problem and solution for reducing the overall costs, which takes the above pre-processed graph as input and recommends retaining a subset of the datasets while deleting the others. Our goal is to minimize the overall cost while maintaining the following constraint. 
We require for every dataset deleted, there should be at least one parent retained. This ensures that the dataset can be reconstructed from existing datasets, if the need arises, within a bounded latency. The total cost we would like to minimize consists of storage costs, data maintenance costs as well as the additional (expected) compute costs for reconstruction of deleted datasets. Let us denote this problem as \textsc{Opt-Ret}.

Now we will formally define the optimization problem. We are given as input a directed graph $\mathcal{G} = \{\mathcal{V}, \mathcal{E}\}$, where $\mathcal{V}$ denotes the set of nodes (datasets), $\mathcal{E}$ denotes the set of directed edges between datasets denoting the containment relation. 
Retaining a node contributes to both storage costs as well as maintenance costs per node (due to data hygiene and other data management operations, including privacy requests and GDPR). 
Let the size of a dataset node $v$ be $S_v$. 
Considering the storage cost per unit size of data for a billing period to be $C_s$, retaining the node contributes $C_s \cdot S_v$ to the storage cost. Let the average (or, expected) compute costs for maintenance per unit size data be $C_m$ (this can be determined from historical data maintenance logs, and corresponding compute costs). Let $f_v$ denote the expected frequency of such maintenance operations for $v$ for a billing period. The total cost of retention of $v$ for the billing period is therefore $(C_s + C_m\cdot f_v)\cdot S_v$. 

Let $\mathcal{P}_v$ denote the set of parents of $v$. From every node $u \in \mathcal{P}_v$, $v$ has an incoming edge and $v$ has incoming edges only from the nodes in $\mathcal{P}_v$. Now, let us consider the cost of deleting a node. If we delete $v$, then there must not be any loss of information, hence at least one parent should be retained. However, as discussed earlier, there would be additional compute costs for the potential reconstruction of datasets. Let the (expected) compute cost of reconstructing $v$ from its parent $u$ through the edge $e = (u,v)$ be $C_e$. (Recall that an edge exists only if the latency constraints are satisfied in expectation.)

Let the expected number of (customer initiated) accesses to dataset $v$ over a billing period (this can be determined from historical logs) be $A_v$. 

If the parent $u$ is used for reconstructing $v$ if accessed, then the cost is given by $C_e \cdot A_v$, where $e = (u,v)$. Let us use $x_v$ as the indicator variable of whether $v$ is retained or not ($x_v = 1$ for retained). Let us use $y_{e = (u,v)}$ as the indicator of whether $u$ is the parent used for reconstruction of $v$ in case of deleting $v$ ($e$ is the directed edge from $u$ to $v$). 
A dataset $v$ can be reconstructed from a parent $u$  only if the parent $u$ is retained. Hence, we need to add a constraint $y_{e = (u,v)} \leq x_u$ for every edge $e = (u,v)$. 
There must be a parent retained for every $v$ deleted. 

Hence we add a constraint: $x_v + \sum_{e = (u,v) \ \forall u \in \mathcal{P}_v}{y_e} \geq 1$.

However, we want $y_e$ to be $1$, only when the child node is deleted, hence, we add a constraint $y_{e = (u,v)} \leq 1 - x_v$. Next, we describe \textsc{Opt-Ret} in Equation \ref{eq:opt} .

\begin{align}
\label{eq:opt}
    &\text{Minimize } \sum_{v \in \mathcal{V}}{(C_s + C_m f_v) S_v x_v} + \sum_{v \in \mathcal{V}}{{\sum_{e = (u,v) \ \forall u \in \mathcal{P}_v}{A_v C_e y_e}}} \\\nonumber
    & \text{s.t. } \\\nonumber
    & y_{e = (u,v)} \leq x_u \quad \forall e=(u,v) \in \mathcal{E}, \quad u \in \mathcal{P}_v, \quad v \in \mathcal{V} \\\nonumber
    & x_v + \sum_{e = (u,v) \ \forall u \in \mathcal{P}_v}{y_e} \geq 1 \quad \forall v \in \mathcal{V}\\\nonumber 
    & y_{e = (u,v)} \leq 1 - x_v \quad \forall e = (u,v) \in \mathcal{E}, \quad \forall v \in \mathcal{V}\\\nonumber
    & x_v \in \{0, 1\} \quad \forall v \in \mathcal{V}\\\nonumber
    & y_e \in \{0, 1\} \quad \forall e \in \mathcal{E}\\\nonumber
\end{align}

\subsection{Linear Algorithm for Line Graph}
The problem \textsc{Opt-Ret} admits a linear time algorithm for the special case when the pruned directed graphs are line graphs. That is, every parent has a single child, and every child has a single parent. This case can often arise in enterprise data lakes, for example, when a sequence of edits or operations are performed starting from a root dataset, saving every intermediate output. This special case is solved by an efficient dynamic program \textsc{Dyn-Lin}. 

We next define the recursive equations for \textsc{Dyn-Lin}. Without loss of generality, we could solve the problem optimally for each line graph (that the input consists of). For simplicity, here we discuss the optimal solution for a single line graph.  Let there be $N$ nodes in the line graph, where node $0$ is the root (that is, no incoming edges) and node $N-1$ is the leaf (that is, no outgoing edges). Let $ALG[i]$ denote the cost of the optimal solution for nodes $[0, \ldots, i]$. 
\begin{align}
\label{eq:dyn-lin}
    ALG[0] &= \left(C_s + C_m\cdot f_0\right)\cdot S_0\\\nonumber
    ALG[1] &= \min{\{\left((C_s + C_m\cdot f_1)\cdot S_1\right), \left(A_1\cdot C_{0, 1}\right)\}} + ALG[0]\\\nonumber
    \forall i &\in \{2, \ldots, N-1\}:\\\nonumber
    ALG[i] &= \min
    \begin{cases}
       \left((C_s + C_m\cdot f_i)\cdot S_i + ALG[i-1]\right), \\
       \left(A_i\cdot C_{i-1, i} + (C_s + C_m\cdot f_{i-1})\cdot S_{i-1} + ALG[i-2]\right)
    \end{cases}
\end{align}

\begin{theorem}
There exists an optimal linear time algorithm ($O(N)$) that finds the optimal cost and the corresponding optimal solution given a directed line graph. 
\end{theorem}
\begin{proof}
The base cases are as follows. $ALG[0]$ is equal to the cost of retaining node $0$, since this node must be retained. There is no parent for node $0$ from which it could be reconstructed. $ALG[1]$ is the greedy choice (minimum cost) between the retention cost of node $1$ versus the deletion (that is, reconstruction) cost of node $1$, because node $0$ is retained by default. 
Consider $ALG[3]$. This is the greedy choice between retaining node $2$ and then adding the cost of the optimal solution till node $1$, or deleting node $2$, therefore compulsorily retaining node $1$, paying its retention cost plus the cost of the solution up to node $0$ (which is basically the retention cost of node $0$). Clearly, this choice is optimal, as any other choice would lead to sub-optimality in the decision, which could be improved by bringing in this greedy choice or violating the parent retention constraint. Now, let us assume by induction hypothesis that the algorithm is optimal for all $k\leq i< N-1$. Now, consider the solution for node $i+1$. It is again the greedy minimum cost choice between the costs of retaining the node $i+1$, plus the optimal cost till $i$, $ALG[i]$, and that of deleting node $i+1$, compulsorily retaining node $i$, paying its retention cost, plus the optimal cost till $i-2$. Since there are only two possibilities at each step, the greedy choice is optimal, and replacing with other choices would either lead to a higher cost that can be lowered by switching to the choice recommended by backtracking \textsc{Dyn-Lin}, or lead to violation of constraints. Therefore, by induction, we prove that \textsc{Dyn-Lin} is optimal. 
The algorithm makes one pass through all the nodes to build the optimal cost solution, and then another pass while backtracking to determine the optimal sequence of nodes to retain, leading to overall $O(N)$ complexity. 
\end{proof}

\section{Experimental Results}
\label{sec:experiments}
\subsection{Datasets}
We evaluate our solution on both enterprise and synthetic data. 
\begin{itemize}
    \item {\textbf{Enterprise Data}} - we used 3 enterprise customer accounts with many datasets, each with data ranging from 0.6TB - 42TB in size.
    \item {\textbf{Synthetic Data}} -  We constructed two sets of synthetic data. One was constructed using an existing public dataset that is commonly used in this domain, namely Table Union Benchmark \cite{tableunion}. We generated synthetic data consisting of around 300 tables, with the entire dataset size being 324MB. 
    The other dataset was created using tables from Kaggle competitions as root tables (refer Section \ref{sec:synthetic}). The synthetically generated dataset consists of 140 tables, with the entire dataset size being 24GB. Further details regarding the data generation process are given below.
\end{itemize}

The enterprise data is partitioned and stored in parquet format in our enterprise Azure data lake (ADLS Gen2\footnote{\url{https://learn.microsoft.com/en-us/azure/storage/blobs/data-lake-storage-introduction}}). Experiments on enterprise data were run using Apache Spark and Azure Databricks on an L16s storage optimized cluster (128GB memory, 16 cores). We queried data and metadata directly from the Azure data lake using Spark, and the reported times of our experiments reflect this. Note that Spark does not have an indexed database. Using an indexed database can potentially reduce the processing time further since it can speed up certain operations like sampling. Experiments on synthetic data were run in jupyter notebooks on a c5.24xlarge AWS CPU cluster (192GB RAM, 96 CPUs) without parallelism. 

\subsubsection{Synthetic Data Generation Process}
\label{sec:synthetic}
In typical data lakes, data is often processed and transformed, and the results are saved as a new table. We started with a set of root datasets from Table Union Data \cite{tableunion} (which has been used in prior work) and Kaggle competitions. We simulate the main types of transformations and processing that occur in real data lakes via the following:

\begin{itemize}
    \item Size reduction via Sampling: We generated synthetic {\verb|SELECT| ... \verb|FROM| ... \verb|WHERE| ...} queries based on a skewed Zipfian distribution whose parameters were fitted based on enterprise queries that followed the same distribution.
    \item Adding rows: The new values are chosen by sampling from each respective column's distribution.
    \item Adding columns: Adding new columns to the tables which are linear combinations of existing numerical columns.
    \item Noise: Adding noise to certain numerical columns.
    \item A combination of the above.
\end{itemize} 

\subsection{Generating Ground Truth}
We created the ground truth for both schema and content level containment in a brute force manner. For each pair of tables, we checked the containment of schema sets to compute the ground truth schema graph. Then for each edge, we checked whether each row of the smaller table occurs in the larger table to compute the ground truth containment graph. 

We compare R2D2's results with several baselines, including the ground truth graphs for both schema and content level containment. The ground truth graphs can be used to check for correctness. We also compare the time and complexity of generating these graphs to understand the extent of improvement we bring.

\subsection{Evaluating Containment Graph}
We compare the graph after each step of the pipeline with the ground truth containment graph for both enterprise and synthetic data in Tables \ref{tab:results-enterprise} and \ref{tab:results-synthetic}, respectively. After every step of our pipeline, we reduce the number of incorrect edges in the graph by a significant number. In Tables \ref{tab:results-enterprise} and \ref{tab:results-synthetic}, the label `Incorrect(<1)' refers to all edges between datasets with a containment fraction less than 1. This is because in our use case, an edge is correct only if the child node is completely contained in the parent.

It is important to note that all the correct edges are captured from the schema step itself, and we focus on reducing the total number of incorrect edges in each step.

\subsection{Comparison with Baselines}
In order to evaluate our method in comparison to other approaches, we compare with several baselines.
To the best of our knowledge, there is no method in the literature that directly computes table level containment, thus we have modified several related approaches to serve as our baselines. We next discuss these approaches for both schema containment as well as content containment.

\begin{table}[htbp]
\resizebox{0.75\linewidth}{!}{%
\begin{tabular}{|c|c|c|c|c|c|}
\hline
\multirow{2}{*}{\textbf{Data}} & \textbf{Size} & \textbf{Number} & \textbf{Graph after} & \textbf{Graph after} & \textbf{Graph after} \\
 & \textbf{(TB)} & \textbf{of edges} & \textbf{SGB} & \textbf{MMP} & \textbf{CLP} \\ \hline
\multirow{4}{*}{Customer 1} & \multirow{4}{*}{0.681} & Correct & 278 & 278 & 278 \\ \cline{3-6} 
 &  & Incorrect (\textless 1) & 6657 & 3414 & 110 \\ \cline{3-6} 
 &  & Not detected & 0 & 0 & 0 \\ \hline\hline
\multirow{4}{*}{Customer 2} & \multirow{4}{*}{41.8} & Correct & 31 & 31 & 31 \\ \cline{3-6} 
 &  & Incorrect (\textless 1) & 1192 & 600 & 315 \\ \cline{3-6} 
 &  & Not detected & 0 & 0 & 0 \\ \hline\hline
 \multirow{4}{*}{Customer 3} & \multirow{4}{*}{27.6} & Correct & 21 & 21 & 21 \\ \cline{3-6} 
 &  & Incorrect (\textless 1) & 1769 & 421 & 272 \\ \cline{3-6} 
 &  & Not detected & 0 & 0 & 0 \\ \hline
\end{tabular}}
\caption{Enterprise data results: Number of correct, incorrect, and undetected edges after each step of our algorithm, with respect to the ground truth containment graph. We are able to preserve the correct edges while significantly reducing the number of incorrectly detected edges in each step of the R2D2 pipeline.}
\label{tab:results-enterprise}
\end{table}

\subsubsection{Schema Containment}
We considered the following baselines. 

\begin{enumerate}

\item Ground Truth Schema - A brute force baseline is to compare the schema sets for all pairs of datasets and check whether containment holds.

\item Bharadwaj et. al \cite{bharadwaj} - It considers tables that occur in the same join query to be "joinable". While we have the access logs at a dataset level, we did not have access to actual query data for the corresponding customer orgs. Moreover, our use case requires looking at containment, not joinability. The ground truth schema containment graph gives us examples of pairs of tables that satisfy the schema containment condition (we call them positive samples). For negative samples, we randomly sample two tables that are not present in the ground truth graph. For every pair of tables, we build the feature vector using column name similarity and column name uniqueness as done in the original paper. Further, we train multiple classifiers on this set of positive and negative samples with the task of predicting whether containment exists. Accuracy numbers are reported in Table \ref{tab:bharadwaj-comparison}. 
    
\item KMeans Clustering - Our method, SGB computes clusters algorithmically. We compare this with a baseline of generating clusters with KMeans/KMedoids. We get embedding vectors for each table schema by computing the average of the column embedding vectors for that table. We then employ KMeans clustering to create schema clusters based on these embedding vectors. Pairwise schema containment is computed for members within each cluster similar to SGB. Evaluation is done with respect to the ground truth schema containment graph.
\end{enumerate}
We compare the generated table pairs of our method as well as the baselines with the ground truth schema containment graph in Table \ref{tab:bharadwaj-comparison}. It can be clearly observed from the above results that both baselines (which are embedding based) perform worse than our deterministic SGB algorithm. 

\begin{table}[htbp]
\resizebox{0.75\linewidth}{!}{%
\begin{tabular}{|c|c|c|c|c|c|}
\hline
\multirow{2}{*}{\textbf{Data}} & \multirow{2}{*}{\textbf{Size}} & \textbf{Number} & \textbf{Graph after} & \textbf{Graph after} & \textbf{Graph after} \\
 & & \textbf{of edges} & \textbf{SGB} & \textbf{MMP} & \textbf{CLP} \\ \hline
\multirow{3}{*}{Table Union} & \multirow{3}{*}{324 MB} & Correct & 1863 & 1863 & 1863 \\ \cline{3-6} 
 &  & Incorrect (<1) & 2902 & 707 & 115 \\ \cline{3-6} 
 &  & Not detected & 0 & 0 & 0 \\ \hline\hline
\multirow{3}{*}{Kaggle} & \multirow{3}{*}{24 GB} & Correct & 1093 & 1093 & 1093 \\ \cline{3-6} 
 &  & Incorrect ( <1) & 1663 & 476 & 58 \\ \cline{3-6}  
 &  & Not detected & 0 & 0 & 0 \\ \hline
\end{tabular}}
\caption{Synthetic data results: Number of correct, incorrect, and undetected edges after each step of our algorithm, with respect to the ground truth containment graph. We are able to preserve the correct edges while significantly reducing the number of incorrectly detected edges in each step of the R2D2 pipeline.}
\label{tab:results-synthetic}
\end{table}

\begin{table*}[htbp]
\resizebox{1\linewidth}{!}{%
\begin{tabular}{|c|c|cccc|}
\hline
\multirow{3}{*}{\textbf{Method}} & \multirow{3}{*}{\textbf{Complexity}} & \multicolumn{4}{c|}{\textbf{Pairwise Operations}} \\
& & \textbf{Customer 2} & \textbf{Customer 1} & \textbf{Kaggle} & \textbf{Table Union} \\
& & 41.8TB & 0.68TB & 24 GB & 324 MB\\ \hline
Ground Truth Schema & ${N\choose 2}$ & $1.34\times 10^9$ & $1.47\times 10^5$ & $9\times 10^3$ &  $4.33\times 10^4$  \\ \hline
\textbf{SGB} ($E_1$ edges) & $N\log{N}+K(N-K)+\sum_{i}^{K} {K_i\choose 2}$ & $2.13\times 10^5$ & $1.32\times 10^5$ & $2.05\times 10^4$ &  $9.1\times 10^4$ \\ \hline\hline
Ground Truth Content & $\sum_{i,j}^{E_1} M_iM_j$ & $7.36\times 10^{21}$ & $7.4\times10^{21}$ & $5.55\times 10^{15}$ & $9.77\times10^{11}$\\ \hline
\textbf{MMP} ($E_2$ edges) & $E_1$ & 1192 & 6657 & 5512 & 9530\\ \hline
\textbf{CLP} & $\sum_i^{E_2} M_it$ & $1.06\times 10^{10}$ & $2.5\times10^{12}$ & $3.39\times10^{10}$ & $4.94\times10^{8}$\\ \hline
\end{tabular}}
\caption{Comparison of the number of pairwise row-level operations that need to be conducted to compute schema and table-level containment. $N$ = number of tables, $M_i$ = number of rows in table $i$, $K_i$ = number of tables in cluster $i$, $t$ = max rows to sample in CLP, $E_i$ = number of edges in the resulting graph. We compare each step of R2D2's pipeline with standard brute force approaches to compute containment. Comparison is done on the same synthetic and enterprise datasets in Tables 1 and 2 at various scales (MBs - TBs).}
\label{tab:row-ops}
\end{table*}

\begin{table}[htbp]
\resizebox{0.45\linewidth}{!}{%
\begin{tabular}{|c|c|c|c|}
\hline
\multirow{2}{*}{\textbf{Data}} &  \multirow{2}{*}{\textbf{Method}} & \textbf{Correctly} & \textbf{Not} \\
 &   & \textbf{Identified} & \textbf{Detected} \\ \hline
\multirow{3}{*}{Customer 1} & \cite{bharadwaj} & 10774  & 363  \\ \cline{2-4}
& KMeans & 8781  & 2356  \\ \cline{2-4} 
& SGB & 11137  & 0  \\
 \hline\hline
\multirow{3}{*}{Customer 2} & \cite{bharadwaj} & 2009 & 61 \\ \cline{2-4}
& KMeans & 1455  & 615  \\ \cline{2-4}
& SGB & 2070  & 0  \\ \hline
\end{tabular}}
\caption{Comparing \cite{bharadwaj}, KMeans, and our approach (SGB) on enterprise data for schema containment. The number of correctly detected and undetected edges are reported with respect to edges in the ground truth schema graph.}
\label{tab:bharadwaj-comparison}
\end{table}

\subsubsection{Full Table Containment}
Next, for full table level containment (assuming the schema step is done), there are a few baselines that solve the problem of set containment and not table containment. Though the baselines cannot be directly applied to our problem setting, we explain below how we modified the baselines and why they fail to produce the desired results in the current scenario.

\textbf{Ground Truth} - This is a standard brute-force approach, where we take each edge (pair of tables) from the schema graph, and subsequently compare hashes of all possible row pairs to check the extent of containment for each edge. This does not scale.

\textbf{LCJoin} \cite{lcjoin} - It solves the problem of finding the subset relationship between sets from two collections of sets. We can apply this in two ways to our setting. In our setting, we can treat columns as sets and then find the subset relationship between all the "column sets". This gives inaccurate results since we've already discussed before that \textit{column containment does not translate to table containment}. Another option is to treat every table as a set, where an element of the set is represented by a row (tuple of values). This also gives inaccurate results since one table might be fully contained in another table with a larger number of columns but LCJoin won't be able to find it since the elements of the sets in both tables will be of different sizes.

\textbf{JOSIE} \cite{josie} - It deals with the problem of finding top-k related sets for a query set. We cannot apply JOSIE directly to our problem setting because of similar reasons as LCJoin. JOSIE, additionally, is concerned with table relatedness and not containment.

\subsection{Evaluating Scalability}

In comparison with ground truth graph generation, R2D2 is several orders of magnitude faster on both enterprise and synthetic data as indicated by the numbers in Table \ref{tab:time-synthetic}. 
While running the pipeline on synthetic data, we cache the minimum and maximum values of all the columns, since these are generally present as part of table/partition level metadata. The comparable times between the Kaggle and enterprise data are due to the differences in parallelism, optimization (present in Spark) and the compute clusters used to run the models.

\begin{table}[htbp]
\resizebox{0.70\linewidth}{!}{%
\begin{tabular}{|c|cc|cc|}
\hline
\multirow{3}{*}{\textbf{Method}} & \multicolumn{4}{c|}{\textbf{Time Taken}} \\ \cline{2-5} 
 & \textbf{Customer 1} & \textbf{Customer 2} & \textbf{Table Union} & \textbf{Kaggle} \\
 & \textbf{0.681 TB} & \textbf{41.8 TB} & \textbf{324 MB} & \textbf{24 GB} \\ \hline
\textbf{Ground Truth} & \multicolumn{1}{c|}{$\sim$days} & $\sim$days & \multicolumn{1}{c|}{3.37 hrs} & $\sim$days \\ \hline
\textbf{SGB} & \multicolumn{1}{c|}{0.51s} & 0.8s & \multicolumn{1}{c|}{0.03 s} & 0.0114 s \\ \hline
\textbf{MMP} & \multicolumn{1}{c|}{7.07 mins} & 8.25 mins & \multicolumn{1}{c|}{9 s} & 3.45 mins \\ \hline
\textbf{CLP} & \multicolumn{1}{c|}{5.12 hrs} & 4.75 hrs & \multicolumn{1}{c|}{5.4 mins} & 5.15 hrs \\ \hline
\textbf{Ours (Total)} & \multicolumn{1}{c|}{$\sim$5.24 hrs} & $\sim$4.89 hrs & \multicolumn{1}{c|}{5.5 mins} & 5.37 hrs \\ \hline
\end{tabular}
}
\caption{Comparison of the time taken for various steps of our optimized pipeline against the ground truth containment computation on enterprise and synthetic data.}
\label{tab:time-synthetic}
\end{table}

In Table \ref{tab:row-ops} we compare the number of row level operations that need to be incurred for computing the ground truth for each step of our pipeline, for various scales of data (MBs - TBs). The smaller size datasets (MBs - GBs) are synthetic, whereas the larger datasets (in TBs) correspond to enterprise data. Note that for the baseline ground truth schema computation, we only need to compare one "row" (that is the flattened column names) from each table.
For the brute force ground truth containment, each row of dataset $A$ needs to be hashed and compared with each row of dataset $B$ corresponding to each edge of the schema graph. This is done for each edge in the schema graph. However, in our pipeline, we prune the edges in an optimized manner, which reduces the required number of operations exponentially.
In addition to the above table, we show a variation of the time taken by our pipeline as the size of the data is varied from 240MB to 41.8TB for Customer 2 and 207MB to 681GB for Customer 1 in Fig. \ref{fig:scalability}. The difference in time between the 2 customers is influenced by the number of edges remaining after the SGB and MMP steps (which is higher for Customer 1). At higher scales, the time is dominated by the CLP step since there are a larger number of rows to scan. However, at lower scales, the MMP step also comes into the picture since the time taken by CLP and MMP become comparable.

\begin{figure} 
    \centering
    \includegraphics[width=0.5\linewidth]{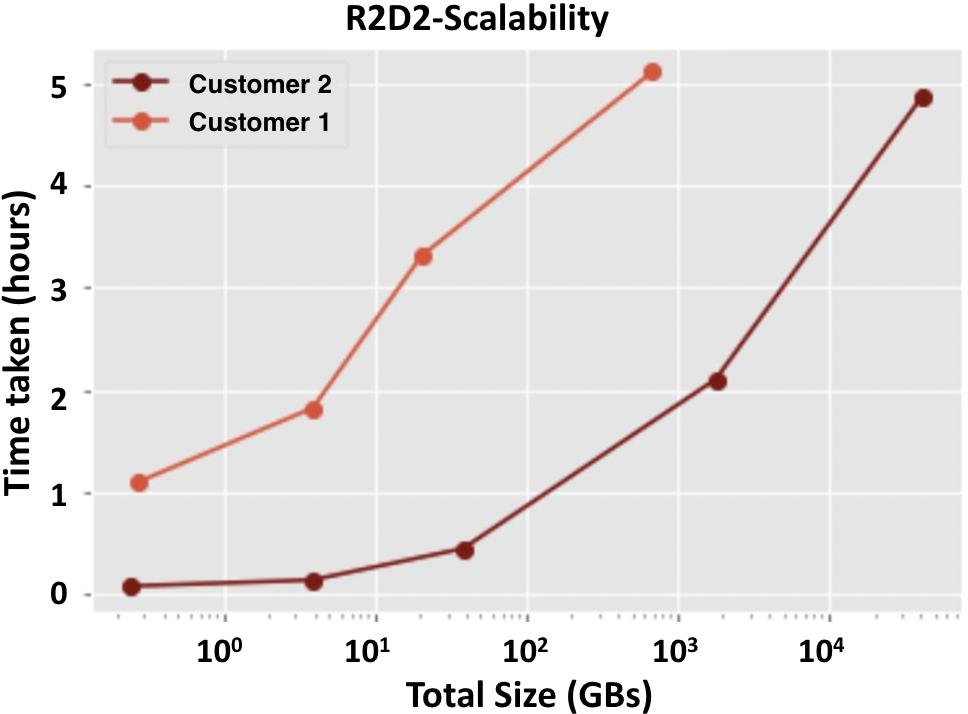}
    \caption{Variation of the time taken by our pipeline (in hrs) vs total size of the data (in GBs) for Customer 1 and 2. For each datapoint on the curve, we ran the pipeline on datasets below a certain size threshold (ranging from 10MB to 10TB).}
    \label{fig:scalability}
\end{figure}

\subsection{Content Level Pruning: Parameter Selection}
The CLP step of our pipeline requires selecting how many rows to sample in addition to the number of columns to sample from. Here we present an analysis of the number of edges reduced in each case while identifying containment for a particular dataset. Empirically, we observed that increasing the number of columns sampled (denoted by $s$) does not help beyond a point (around $s=4$). Similarly, increasing the number of rows sampled (denoted by $t$) results in diminishing returns beyond a point, considering that sampling additional rows and columns takes up extra time. We found that using $s=4, t=10$ is a fairly reasonable parameter configuration to balance speed and accuracy.

Note that in Apache Spark, sampling a table naively results in a full table scan. Even this is significantly faster than comparing rows pairwise, however in such situations, the time taken in the CLP step for different parameter configurations will be comparable. This is because the whole table will be scanned regardless of how the samples are filtered. However, for indexed databases, or situations where the data is partitioned by a particular column like timestamp, the query can run in an optimized manner and only access a subset of rows, which results in significantly lower time taken.

\subsection{Evaluating Optimization Framework}
\label{sec:optim-results}
So far we have looked into the performance of the containment graph generation portion of our pipeline. In this section, we evaluate the optimization framework that takes this graph as input and recommends datasets for deletion based on dataset containment, access patterns, and cloud storage cost parameters.

We used Azure Data Lake Gen2 public hot tier storage and read costs\footnote{\label{footnoteazure}https://azure.microsoft.com/en-in/pricing/details/storage/data-lake/} for the variables $C_s$ and $C_m$ in the integer linear program (ILP) in Eq. \ref{eq:opt}. The expected number of accesses $A$ and maintenance frequency $f_m$ were chosen based on real access patterns for the enterprise datasets. For synthetic data, we sampled $A$ and $f_m$ from a power law distribution. We took the dataset reconstruction cost for an edge $u\rightarrow v$ to be the write cost for writing dataset $v$. Again, the write cost per unit size was taken from the Azure public write costs. 

\begin{table}[htbp]
\resizebox{0.25\linewidth}{!}{%
\begin{tabular}{|c|c|c|c|}
\hline
\textbf{s/t} & \textbf{5} & \textbf{10} & \textbf{30} \\ \hline
\textbf{1} & 908 & 824 & 712 \\ \hline
\textbf{4} & 141 & 122 & 110 \\ \hline
\textbf{8} & 135 & 121 & 109 \\ \hline
\end{tabular}
}
\caption{Comparison of the number of incorrect edges remaining for different parameter configurations in CLP for our 42TB enterprise dataset. Here $s$ (1,4,8) denotes the number of columns that are sampled and $t$ (5,10,30) denotes the number of rows that are sampled from the smaller table.}
\label{tab:clp-params}
\end{table}

We ran two types of experiments to check (i) the number of nodes and edges being deleted and retained, which corresponds to cost reduction; (ii) the optimization scalability and robustness as the number of nodes and edges in the containment graph increase. 

\begin{table}[htbp]
\resizebox{0.75\linewidth}{!}{%
\begin{tabular}{|c|cc|cc|c|}
\hline
\multirow{2}{*}{\textbf{Data}} & \multicolumn{2}{c|}{\textbf{Deletion}} & \multicolumn{2}{c|}{\textbf{Retention}} & \textbf{GDPR Savings} \\ \cline{2-6} 
 & \multicolumn{1}{c|}{\textbf{Nodes}} & \textbf{Edges} & \multicolumn{1}{c|}{\textbf{Nodes}} & \textbf{Edges} & \textbf{(per month, row scans)} \\ \hline
\textbf{Customer 1} & \multicolumn{1}{c|}{19} & 259 & \multicolumn{1}{c|}{99} & 19 & 5.3M \\ \hline
\textbf{Customer 2} & \multicolumn{1}{c|}{12} & 19 & \multicolumn{1}{c|}{18} & 12 & 0.2M \\ \hline
\end{tabular}}
\caption{Optimization results on enterprise data. Given a customer, this illustrates the fraction of datasets we delete and retain in our data lake, along with the extent of compute savings (assuming 1 privacy-initiated access per week).}
\label{tab:optim-enterprise}
\end{table}

The optimization routine accounts for an expected number of accesses, and only recommends deletion if the overall cost benefit is positive. If a dataset is accessed after being deleted, we would need to pay the reconstruction cost involved. This would reduce the overall cost benefit. In Table \ref{tab:optim-enterprise}, there were no extra accesses for the datasets we recommended for deletion. 

We have quantified the overall storage and compute benefits our approach would bring after accounting for such ``accesses after deletion" for a data lake of size 10 PB across a horizon of one year in Fig. \ref{fig:storage-compute}. Savings are computed in terms of storage cost as well as maintenance cost. We consider two cases - (i) 1 privacy-initiated read access per week, or (ii) 5 privacy-initiated accesses per week. We subtract the required read and write costs in case reconstruction is necessary. We used the same public Azure storage, read, and write costs as mentioned earlier.

\begin{figure}[htbp]
    \centering
    \includegraphics[width=0.5\linewidth]{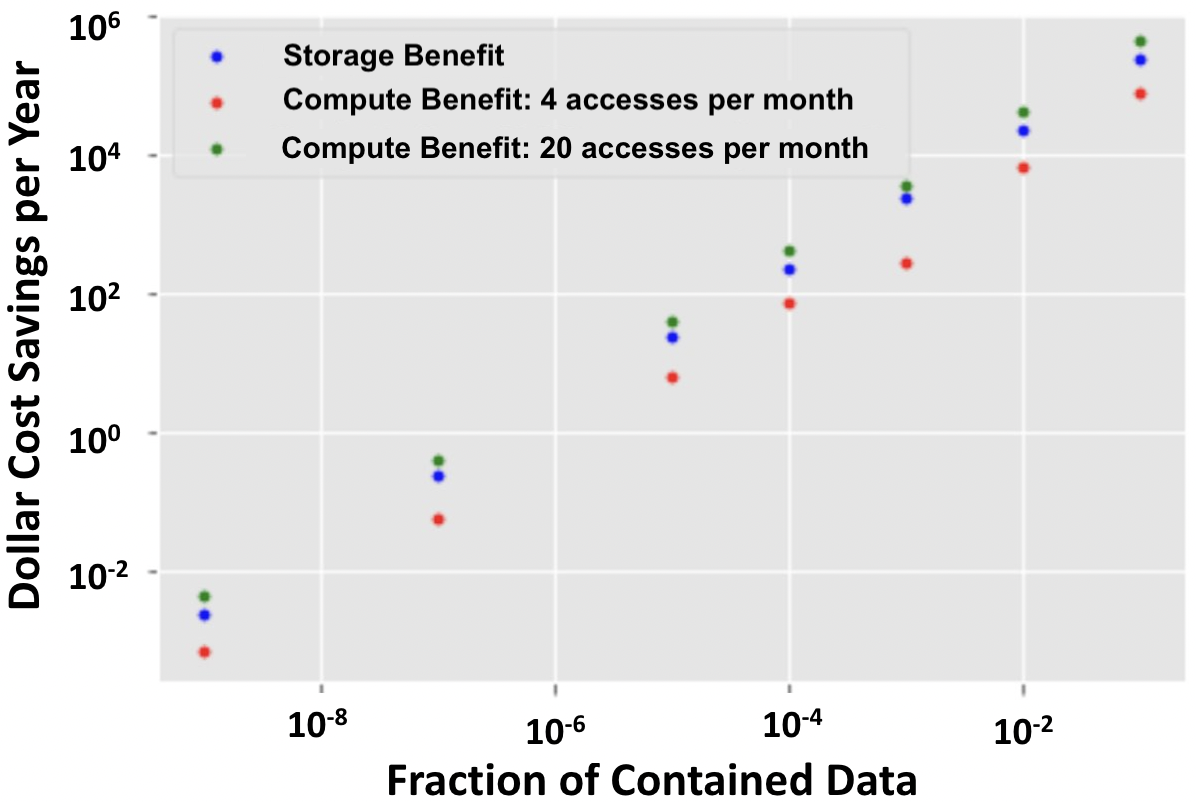}
    \caption{Storage and Compute Cost Benefits for a 10 PB data lake varying with the fraction of contained data across a 1 year horizon. This accounts for any reconstruction write and read cost due to access after deletion.}
    \label{fig:storage-compute}
\end{figure}

We also checked the scalability of the optimization routine by generating random graphs of various sparsity using the Erdos-Renyi model\footnote{Refer \url{https://networkx.org/documentation/stable/reference/generated/networkx.generators.random_graphs.erdos_renyi_graph.html.}} \cite{erdos} (Figure \ref{fig:optimization-scale}). In the first graph, we plotted how the time taken to solve the optimization problem varies as the number of nodes (datasets) increase (while keeping the probability `p' in the Erdos-Renyi model fixed). In the second graph, we plotted how the time taken to solve the optimization problem varies as the number of containment relationships or edges increase (this is varied by changing the probability `p' in the Erdos-Renyi model) while keeping the number of nodes (datasets) fixed.

\begin{figure}[htbp]
  \centering
  \begin{minipage}{0.48\linewidth}
    \centering
    \includegraphics[width=\linewidth]{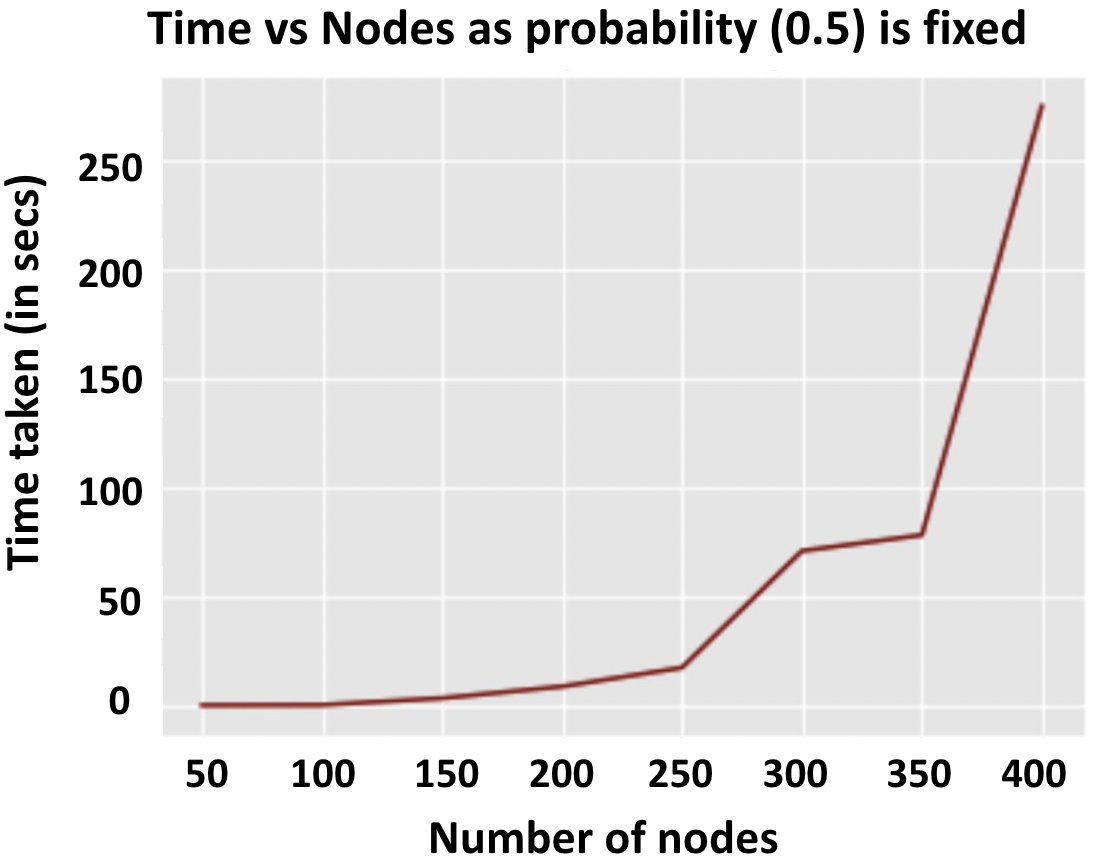}
  \end{minipage}
  \hfill
  \begin{minipage}{0.48\linewidth}
    \centering
    \includegraphics[width=\linewidth]{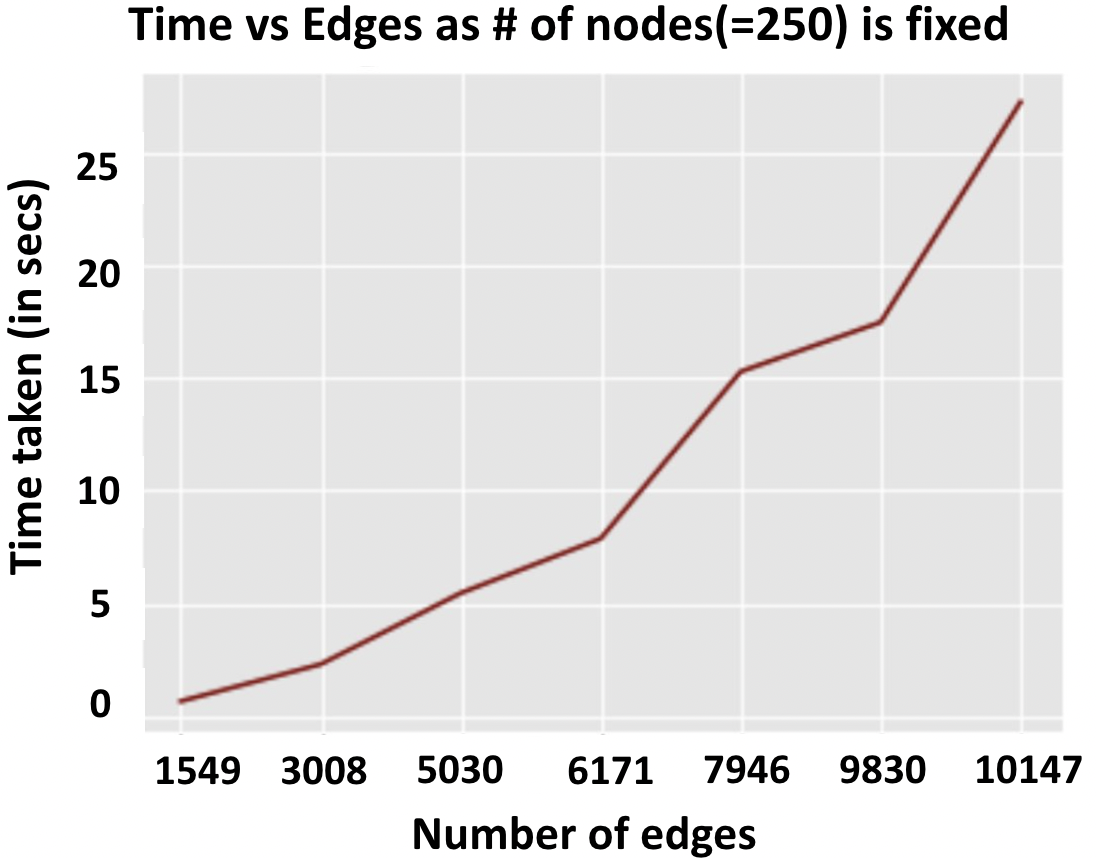}
  \end{minipage}
  \caption{Time taken by the optimization framework as the number of (i) nodes increase and (ii) edges increase.}
  \label{fig:optimization-scale}
\end{figure}

\section{Discussions}
\subsection{Dynamic Graph Updates}
\label{sec:dynamic-updates}
In some enterprise data lakes, it is not possible to directly edit the data. In such cases, only data deletion needs to be considered. However, in others, data can be updated. In such cases, our containment graph needs to be able to be updated as well. One way to do this is by simply running our pipeline again. We propose running the system at a monthly frequency on all the candidate datasets. The time taken to run our system on enterprise data is fairly short even in Apache Spark (a few hours), so this is feasible. If the use case simply requires updating the containment graph on the fly, we can do that efficiently. We consider the following dynamic update cases. 

\textbf{Adding new datasets: }
If a new dataset $v$ is added, we can simply check containment between that dataset and all of the others by applying SGB, MMP, and CLP successively. This can be done as follows. First, we check whether $v$ is contained in any of the cluster centers in the SGB schema graph. If yes, we add it to the respective clusters as a member. If not, it is a new cluster center and cluster membership has to be checked for all the datasets, which is linear in the total number of datasets in the graph. We add edges between $v$ and the cluster members, based on containment evaluation. Then, we successively prune the edges incident on $v$ by applying MMP and CLP. The complexity will be \textbf{linear} in the total number of datasets in the graph, which is fast. 

\textbf{Rows or columns added to existing datasets: }
For rows or columns added to a dataset, in the containment graph, all outgoing edges from that dataset will remain. However, any incoming edges as well as relationships with datasets that did not have an edge previously need to be checked. Once again, this is \textbf{linear} in the number of datasets in the graph. 

\textbf{Rows or columns removed from a dataset: }
If rows or columns are removed from a dataset, the incoming edges incident on the dataset in the containment graph remain. However, for outgoing edges, we need to re-check containment. This is also linear in the number of datasets. 

\textbf{Deleting existing datasets: }
While this can be handled easily in the final containment graph by simply removing the corresponding nodes and all the incident edges to and from the deleted nodes, we need to be careful here because it has implications for the optimization routine as well, which is true for the most of the above cases as well. 
Hence, if the use case only requires an updated containment graph, the above steps can be implemented very efficiently and dynamically, but for the full pipeline, including the optimization routine for recommending dataset deletion, it is better to run R2D2 periodically on the entire data lake for accurate optimization. 

\subsection{Approximate Dataset Relatedness}
\label{sec:approx-containment}
In this paper, we have focused on identifying pairs of tables where $CM(P, Q) = 1$, where $n(Q) \geq n(P)$, as defined in Sec. \ref{sec:problem-statement}. A related problem is to identify cases where $CM(P, Q) > T$, where $T < 1$. This can be applied at a schema level as well as at a content level. In this situation, the rows and schema of P may not be fully contained within those of Q, since P may have some additional rows and columns, though the majority of the schema or content may be common to both tables.
Approximate Containment is a non-trivial problem and is out of scope for this paper. We next discuss some aspects of approximate containment and the associated challenges.

\subsubsection{Approximate Schema Containment}
It is possible that the schema of two tables are slightly different while their contents are very similar or completely contained within one another. We would treat such tables to have no containment, since we focus on finding both schema and content level containment in an exact manner. However, a notion of approximate containment may be useful - for example, "phone number" may be represented as "Phone", "Mobile", "Work Phone", ... , etc. Spelling errors and abbreviations are also possible. It would be useful to have a method to identify whether two tokens represent the same type of content just by looking at the schema. If a canonical list of possible schema tokens exists beforehand, we can identify which tokens convey the same meaning and map them to the same value by maintaining such a look-up or through human input. Such cases of containment can be handled by SGB by mapping the schema tokens to canonical values. However, computing this mapping automatically is nontrivial, since different datasets can have different meanings and quirks based on how the data has been processed. For instance, "Work phone" and "home phone" are both phone numbers, but we cannot treat these as the same. Often, enterprise schemas contain columns that may look similar but actually represent very different things, e.g. there may be columns "company.product.var0" and "company.product.var1". Understanding the meaning of such columns is limited by our understanding of domain knowledge. Embedding based approaches to cluster schema tokens can easily treat such columns as the same, which would give inaccurate results. Understanding approximate containment between enterprise schemas without knowledge of what possible tokens could exist beforehand is a challenging problem that we leave for future work.

\subsubsection{Approximate Content Containment}
Solving the approximate schema containment problem is a necessary but insufficient condition to solve the approximate content level containment problem. Thus, any mistakes made in the approximate schema step may propagate later on in the pipeline when we start analyzing table content. Additionally, using min-max pruning (MMP) will not work here, since the relative minimum or maximum of two columns are not indicative of their extent of approximate containment. The content-level pruning (CLP) step of our pipeline can be used to query rows from one or both tables and claim approximate containment with a certain degree of confidence. However, care needs to be taken to ensure we do not (i) incur a large number of row level operations and/or (ii) miss cases of approximate containment between tables.

\section{Conclusion}
We present R2D2, a framework to reduce data redundancy and duplication in large enterprise data lakes. It is scalable and substantially better than baselines. R2D2 performs well across different types and scales of data, and can be applied practically in an enterprise setting. Going forward, we want to quantify approximate data relatedness in addition to exact containment, and detect certain types of transformations between tables without human input. 

\bibliographystyle{ACM-Reference-Format}
\bibliography{cameraready}

% \printbibliography

\end{document}